%% file: Paper.tex
\documentclass{sig-alternate}
\usepackage{algorithm}
\usepackage{algorithmic}
\usepackage{amsmath}
\usepackage{amssymb}
\usepackage{bbm}
\usepackage{bm}
\usepackage{color}
\usepackage{dsfont}
\usepackage{graphicx}
\usepackage{subfigure}
\usepackage{times}
\usepackage{url}

\newcommand{\commentout}[1]{}
\newcommand{\junk}[1]{}

\newcommand{\mmr}{{\tt MMR}}
\newcommand{\dum}{{\tt DUM}}

\newtheorem{lemma}{Lemma}

\renewcommand{\qed}{\rule{5pt}{5pt}}

\newcommand{\bw}{{\bf w}}

\newcommand{\bx}{{\bf x}}

\newcommand{\cT}{\mathcal{T}}

\newcommand{\realset}{\mathbb{R}}

\newcommand{\abs}[1]{\left|#1\right|}

\newcommand{\I}[1]{\mathds{1} \! \left\{#1\right\}}

\newcommand{\set}[1]{\left\{#1\right\}}

\DeclareMathOperator*{\argmax}{arg\,max\,}

\begin{document}

\title{DUM: Diversity-Weighted Utility Maximization for Recommendations}

\author{
{Azin Ashkan{\small $~^{*}$},
Branislav Kveton{\small $~^{*}$},
Shlomo Berkovsky{\small $~^{**}$},
Zheng Wen{\small $~^{***}$}}
\vspace{1.6mm} \\
\fontsize{10}{10}\selectfont\rmfamily\itshape
$~^{*}$Technicolor, United States \\
\fontsize{9}{9}\selectfont\ttfamily\upshape
\{azin.ashkan,branislav.kveton\}@technicolor.com
\vspace{1.2mm} \\
\fontsize{10}{10}\selectfont\rmfamily\itshape
$~^{**}$CSIRO, Australia \\
\fontsize{9}{9}\selectfont\ttfamily\upshape
shlomo.berkovsky@csiro.au
\vspace{1.2mm} \\
\fontsize{10}{10}\selectfont\rmfamily\itshape
$~^{***}$Yahoo Labs, United States \\
\fontsize{9}{9}\selectfont\ttfamily\upshape
zhengwen@yahoo-inc.com}

\maketitle

\begin{abstract}
The need for diversification of recommendation lists manifests in a number of recommender systems use cases. However, an increase in diversity may undermine the utility of the recommendations, as relevant items in the list may be replaced by more diverse ones. In this work we propose a novel method for maximizing the utility of the recommended items subject to the diversity of user's tastes, and show that an optimal solution to this problem can be found greedily. We evaluate the proposed method in two online user studies as well as in an offline analysis incorporating a number of evaluation metrics. The results of evaluations show the superiority of our method over a number of baselines.
\end{abstract}


\input{Introduction}

\input{RelatedWork}

\input{Motivation}

\input{Model}
\input{Experiments}
\input{UserStudy1}

\input{UserStudy2}
\input{Offline}

\input{Conclusion}

\bibliographystyle{abbrv}
\bibliography{paper}

\end{document}

%% file: Introduction.tex

\section{Introduction}
\label{sec:introduction}

The popularity of recommender systems has soared in the recent years. They are widely used in social networks, entertainment, eCommerce, Web search, and many other online services \cite{rsh2011}. Recommenders deal with the information overload problem and select items on behalf of their users. Typically, a recommender scores recommendable items according to their match to the user's preferences and interests, as encapsulated in the user profiles, and then recommends a list of top-scoring items.

A na\"{i}ve selection of top-scoring items may, however, yield a sub-optimal recommendation list. For instance, collaborative  recommenders tend to recommend most popular items appearing in the profiles of numerous users \cite{KorenB11}. While being good recommendations on their own, these items are likely to be known to the user and bear little value. Likewise, content-based recommenders may target user's favorite topics and recommend homogeneous lists that overlook other potentially interesting topics \cite{LopsGS11}. This has brought to the fore the problem of diversity in recommender systems, which has been studied in a number of works \cite{castells2011novelty,halvey2009diversity, mcnee2006being,vargas2011rank,ziegler2005improving}.

In a nutshell, the diversity problem deals with the construction of recommendation lists that cover as wide as possible range of topics of interest. The problem is particularly acute for users with eclectic interests, having no single dominant topic but rather interested in multiple topics. In this case, it is important for the recommendation list to include items that touch upon many topics, in order to increase the chance of answering the current user's need. Repercussions of the diversity problem can be recognized also in other recommender system use cases. Consider the group recommendation problem in heterogeneous (in terms of interests) groups. Another example of the need for diversity is in sequential recommendations, like in the music domain. In both cases, the recommendation list should incorporate diverse items that either appeal to a number of group members or represent a number of music genres~\cite{zhang2012auralist}.

The need for diversity manifests itself also beyond recommender systems. Consider an ambiguous Web search query. Having no knowledge about the context of the query, a search engine may present results pertaining to different interpretations of the query, so that the user can pick the desired one and reformulate the query~\cite{capannini2011efficient}. Another instance comes from text summarization. Unless the desired topic of the summary is known, it should include references to as many aspects of the original document as possible. Also, diversification may be useful in computer supported collaborative work. There, formation of virtual groups may need to bring together users with complementary skills and expertise areas, such that the diversity of the group is important.

In all the above diversification use cases, it is of paramount importance to maintain the trade-off between increasing the list diversity and maintaining the utility of the results~\cite{carbonell1998use, jannach2013recommenders, zhou2010solving}. Diversity typically comes at the account of decreasing the relevance of items, as relevant but redundant items are substituted with less relevant but more diverse ones. Hence, there is a need to strike the balance between the two objectives~\cite{carbonell1998use}, a modular relevance function and a submodular diversity function, for which an approximation to the optimal solution can be computed greedily \cite{nemhauser78approximation}.

In this work, we introduce a different objective diversification function and show that an optimal solution to the diversity problem can be found greedily. We propose a parameter-free method, denoted as \textit{diversity-weighted utility maximization} ($\dum$), which maximizes the utility of the items recommended to users, subject to the diversity of their interests. We cast this problem as finding the maximum of a modular function subject to a submodular constraint~\cite{edmonds70submodular}, which is known to have an optimal greedy solution. This solution guarantees that items in the recommendation list cover different interests in the user profile, such that each topic of interest is represented by items with high utility. In other words, the utility of items remains the primary concern, but it is subjective to maintaining the diversity and avoiding redundancy in the list. We discuss several interpretations of $\dum$ and identify suitable submodular diversity functions.

We conduct an extensive evaluation of the proposed approach. We present two online studies using crowdsourcing, which compare the perceived quality of the lists generated by $\dum$ with baseline settings maximizing a linear combination of utility and diversity. The results show the superiority of the lists generated by $\dum$ over the baseline methods and we characterize the cases when this superiority is prominent. We also present an offline evaluation that applies a variety of metrics to (a) exemplify the trade-off between diversity and utility in recommendations; and (b) demonstrate that $\dum$ successfully outperforms the baselines. Overall, our analyses demonstrate that $\dum$ can effectively deliver personalized recommendations with high degree of utility and diversity, while not requiring a-priori parameter tuning.

In summary, the contribution of this work is two-fold. Firstly, we propose a parameter-free and computationally efficient method aimed at improving the diversity of the recommendation lists, while maintaining their utility. Secondly, we present experimental evaluations -- online user studies and offline experiments alike -- that demonstrate solid empirical evidence supporting the validity of the proposed approach.

\textbf{Note:} The following notation is used throughout this paper. Let $A$ and $B$ be sets, and $e$ be an element of a set. We use $A + e$ instead of $A \cup \set{e}$ and $A + B$ instead of $A \cup B$. Furthermore, we use $A - e$ instead of $A \setminus \set{e}$ and $A - B$ instead of $A \setminus B$. We represent \emph{ordered sets} by vectors and also refer to them as \emph{lists}.

%% file: RelatedWork.tex

\section{Related Work}
\label{sec:related work}


A common approximation to diversified ranking is based on the notion of \textit{maximal marginal relevance ($\mmr$)} proposed by Carbonell and Goldstein~\cite{carbonell1998use}. In this approach, utility (e.g., relevance) and diversity are represented by independent metrics. Marginal relevance is defined as a weighted combination of these two metrics, to account for the trade-off between utility and diversity. Given a standard ranking of items, $R$, a diversified re-ranking of these items, $S$, is created (Algorithm~\ref{alg:MMR}). In each iteration, an item $e^\ast \in R - S$ is chosen, such that it maximizes the marginal relevance: 
\begin{align}
  e^\ast = \argmax_{e \in R - S} (1-\lambda) \bw(e) + \lambda f(S + e)
\label{eq:mmr}
\end{align}
where $\bw(.)$ and $f(.)$ represent the notions of utility and diversity, respectively, and the parameter $\lambda$ controls the trade-off between the two. Typically, the utility $\bw$ is a modular function of $S$, whereas the diversity $f$ is a submodular function of $S$. The existing approaches differ in how they account for different aspects of query or user (or any other entity of interest) to model $f(.)$. 

\begin{algorithm}[!ht]
  \caption{$\mmr$: Maximal Marginal Relevance}
  \label{alg:MMR}
  \begin{algorithmic}
    \STATE {\bf Input:} 
    \STATE \quad Standard ordering of items $R$
    \STATE \vspace{-0.05in}
    \STATE $S \leftarrow (), n=|R|$
    \WHILE{$|S| < n$}
        \STATE $\displaystyle e^{\ast} \leftarrow \argmax_{e \in R - S}
        \lambda \bw(e) + (1-\lambda) f(S + e)$
        \STATE $R \leftarrow R - e^{\ast}$
        \STATE Append item $e^{\ast}$ to list $S$ 
    \ENDWHILE
     \STATE \vspace{-0.05in}
   \STATE {\bf Output:}
    \STATE \quad List of recommended items $S$
  \end{algorithmic}
\end{algorithm}

Implicit approaches assume that similar items should be penalized since they cover similar aspects. For instance, Yu et al.~\cite{yu2014latent} compute $f(S + e) = -\max_{d' \in S} \mathrm{Sim}(e, e')$ to measure the redundancy of user intent $e$ with respect to a set of selected intents $S$, where $\mathrm{Sim}(e, e')$ is the cosine similarity between the user intents $e$ and $e'$. Gollapudi and Sharma~\cite{gollapudi2009axiomatic} propose multiple diversification objectives considering the tradeoff between relevance and diversity and using various axioms, relevance functions, and distance functions. Their distance functions are defined based on various implicit metrics, e.g., document content, in order to capture the pairwise similarity between any pair of documents. 

On the other hand, explicit approaches model different aspects (e.g., query intent, query topic, or movie genre) directly, and promote diversity by maximizing the coverage of selected items with respect to these aspects. For instance, Santos et al.~\cite{santos2010exploiting} define $f(S + e) = \Sigma_{t \in \cT_q} P(t | q) P(e, \bar{S} | t)$ where $P(d, \bar{S} | t)$ represents the likelihood of document $e$ satisfying topic $t$ while the ones in $S$ fail to do so. Also, $P(t | q)$ denotes the popularity of $t$ among all possible topics $\cT_q$ that may satisfy a user's information need from issuing query $q$.

In addition to the above approaches in diversifying existing rankings, another group of work directly learns a diverse ranking by maximizing a submodular objective function. Among these approaches, Radlinski et al.~\cite{radlinski2008learning} and Yue and Guestrin~\cite{yue2011linear} propose \textit{online} learning algorithms for optimizing a class of submodular objective functions for diversified retrieval and recommendation, respectively. Agrawal et al.~\cite{agrawal2009diversifying}, on the other hand, address search result diversification in an offline setting, with respect to the topical categories of documents. The authors target the maximization of a submodular objective function following the definition of marginal relevance. 
They propose a greedy algorithm to approximate the objective function and show that an optimal solution can be found in a special case, where each document belongs to exactly one category.

Vallet and Castells~\cite{vallet2012personalized} study personalization in combination with diversity such that the two objectives complement each other in addressing various query aspects and satisfying user needs. In particular, they generalize the work of Agrawal et al.~\cite{agrawal2009diversifying} and Santos et al.~\cite{santos2010exploiting} to the personalized versions by exploiting available information about user preferences.



Most of these studies target diversity in information retrieval, while there has been a growing interest in recommendation diversification more recently. One of the initial works in recommendation diversification is by Ziegler et al.~\cite{ziegler2005improving}, who argue that user satisfaction does not solely depend on the accuracy of recommendation results. The authors propose a similarity metric, the intra-list similarity (ILS), which computes the average pairwise similarity of items in a list. A higher value of the metric denotes a lower diversity. They use this metric in their topic diversification model to control a balance between the accuracy and diversity of recommendations.

Zhang et al.~\cite{zhang2008avoiding} formulate the diversification problem as finding the best possible subset of items to be recommended over all possible subsets. They address this as the maximization of the diversity of a list of recommended items, subject to maintaining the accuracy of the items. Zhou et al.~\cite{zhou2010solving} propose a hybrid method that targets the maximization of a weighted combination of independent utility- and diversity-based approaches requiring parameter tuning to control the tradeoff between the objectives of the two approaches.

Most of the existing diversification approaches consider the maximization of an objective function to satisfy a user's need in terms of utility and diversity of the result list. Most of these approaches are based on the idea behind the maximal marginal relevance where a submodular objective function (Equation~\ref{eq:mmr}) is maximized. Therefore, a $(1-1/e)$-approximation to the optimal solution can be computed greedily~\cite{nemhauser78approximation}. This paper is an extension to our prior work in~\cite{ashkan2014diversified}, where we introduce a new objective function for recommendation diversification, the optimal solution of which can be found greedily. This objective function targets the utility as the primary concern, and maximizes it subject to maintaining the diversity of user's tastes. In this paper, we elaborate on the intuitions and details behind the proposed greedy algorithm, and provide extensive online and offline evaluations on its performance in practice. We show that this method is computationally efficient and parameter-free, and it guarantees that high-utility items appear at the top of the recommendation list, as long as they contribute to the diversity of the list. 






%% file: Motivation.tex

\section{Motivating Examples}
\label{sec:motivation}

In this section, we discuss several motivating examples for our work. A more formal description of our method and its analysis are presented in Section~\ref{sec:model}.

Consider the following recommendation problem. A system recommends to a user movies from a ground set of four movies:
\begin{center}
  \begin{tabular}{ccccc} \hline
    ID & Movie & Utility & Action & Comedy \\
    $e$ & name & $\bw(e)$ & & \\ \hline
    1 & Inception & 0.8 & X & \\
    2 & Spider-Man 2 & 0.7 & X & \\
    3 & Grown Ups 2 & 0.5 & & X \\
    4 & The Sweep & 0.2 & & X \\ \hline
  \end{tabular}
\end{center}
The user likes either \emph{Action} or \emph{Comedy} movies, depending on the mood of the user, but the system does not know the user's mood. The user chooses the first recommended movie $e$ that matches the genre that the user currently prefers and is satisfied proportionally to the utility of the movie $\bw(e)$, the probability that $e$ is liked. Our goal is to recommend a minimal list of movies that maximizes the user's satisfaction and also covers all user's preferences, irrespective of the user's mood.

The optimal solution to our problem is a list of two movies, $S = (1,3)$. When the user prefers \emph{Action} movies, the user selects the first recommended movie in the list, \emph{Inception}, and is satisfied with probability $0.8$. This is substantially greater than if \emph{Spider-Man 2}, another \emph{Action} movie, was in the list instead of \emph{Inception}. On the other hand, when the user prefers \emph{Comedy} movies, the user selects the second recommended movie in the list, \emph{Grown Ups 2}, and is satisfied with probability $0.5$. This is substantially greater than if \emph{The Sweep}, another \emph{Comedy} movie, was in the list instead of \emph{Grown Ups 2}. Note that the solution $S$ can be computed greedily. In particular, $S$ is a list of two highest-utility movies, one from each genre.

Now suppose that we add to the ground set a movie that is both \emph{Action} and \emph{Comedy}, and its utility is $0.7$:
\begin{center}
  \begin{tabular}{ccccc} \hline
    ID & Movie & Utility & Action & Comedy \\
    $e$ & name & $\bw(e)$ & & \\ \hline
    1 & Inception & 0.8 & X & \\
    2 & Spider-Man 2 & 0.7 & X & \\
    3 & Grown Ups 2 & 0.5 & & X \\
    4 & The Sweep & 0.2 & & X \\
    5 & Kindergarten Cop & 0.6 & X & X \\ \hline
  \end{tabular}
\end{center}
The optimal solution to the problem is a list $S = (1, 5)$. When the user prefers \emph{Action} movies, the user selects the first recommended movie, \emph{Inception}, and is satisfied with probability $0.8$. On the other hand, when the user prefers \emph{Comedy} movies, the user selects the second recommended movie, \emph{Kindergarten Cop}, and is satisfied with probability $0.6$. Note again that the solution $S$ can be computed greedily. It is a list of two highest-utility movies, one from each genre.

Finally, we replace the last movie in the ground set with a movie whose utility is 0.9:
\begin{center}
  \begin{tabular}{ccccc} \hline
    ID & Movie & Utility & Action & Comedy \\
    $e$ & name & $\bw(e)$ & & \\ \hline
    1 & Inception & 0.8 & X & \\
    2 & Spider-Man 2 & 0.7 & X & \\
    3 & Grown Ups 2 & 0.5 & & X \\
    4 & The Sweep & 0.2 & & X \\
    5 & Indiana Jones and  & 0.9 & X & X \\ 
    & the Last Crusade &  &  &  \\ \hline
  \end{tabular}
\end{center}
The optimal solution to the problem is a single movie, $S = (5)$. The reason is that \emph{Indiana Jones and the Last Crusade} is the highest-utility movie in both \emph{Action} and \emph{Comedy}. Hence, it is the best recommendation irrespective of the user's mood. Note again that the solution $S$ can be computed greedily. It is the highest-utility movie that belongs to both genres

In all three examples, the optimal solutions can be computed greedily. This is not by chance. In the next section, we generalize the ideas expressed in these examples and introduce the notion of diverse recommendations where the optimal solution can be found greedily. This is the main contribution of our paper.

%% file: Model.tex

\section{Diversity-Weighted Utility Maximization}
\label{sec:model}

Our objective is to maximize the utility of recommending a list of items to a user subject to the diversity of their tastes. We present the formal definition of our method in Section~\ref{sec:problem}, followed by Section~\ref{sec:algorithm} where we show that the optimal solution of the method can be found efficiently. The interpretations and intuitions behind our method are explained in Section~\ref{sec:interpretation}. We show that the length of the list recommended by our method can be controlled by considering different diversity constraints dependent on user preferences in Section~\ref{sec:personalization}.

\subsection{Problem Formulation}
\label{sec:problem}

Let $E = \set{1, \dots, L}$ be a ground set of $L$ recommendable items, such as movies or songs. Let $\bw \in (\realset^+)^L$ be a vector of item utilities, such as item popularity scores or predicted ratings. The $e$-th entry of $\bw$, $\bw(e)$, is the utility of item $e$.

The objective of the diversification method is to maximize the satisfaction of the user subject to the diversity of their tastes. However, an increase in diversity typically comes at the account of a decrease in the utility of the items in the list, e.g., relevant but redundant items are substituted by less relevant but more diverse items. Addressing this tradeoff and striking the balance between increasing the diversity and maintaining the utility is an important challenge for any diversification method. Considering the \textit{utility} as the primary concern, we aim to expos the user to a variety of \textit{choices} in the recommendation list, while losing the minimal amount of utility in the provision of these choices.

In order to recommend a list of items that maximizes the user's utility of choice, we target at maximizing the utility of the recommendation list weighted by the increase in diversity. In other words, each increase in diversity is covered by the item with the highest possible utility. Formally, our diversity-weighted utility maximization ($\dum$) problem is formulated as:
\begin{align}
  A^\ast = \argmax_{A \in \Theta} \sum_{k = 1}^L g_A(a_k) \bw(a_k),
  \label{eq:DUM}
\end{align}
where $A = (a_1, \dots, a_L)$ is an ordered set of items $E$, $\Theta$ is the set of all permutations of $E$, and $A^\ast = (a^\ast_1, \dots, a^\ast_L)$ is the optimal solution to the problem. The vector $g_A \in (\realset^+)^L$ are the gains in diversity associated with items $E$. In particular:
\begin{align}
  g_A(e) = f(A_{k - 1} + e) - f(A_{k - 1})
  \label{eq:gain}
\end{align}
is the gain in diversity associated with choosing item $e$ given a set of previously chosen items in $A$, where $k$ is such that $a_k = e$ and $A_k = \set{a_1, \dots, a_k}$ is an unordered set of the first $k$ items in $A$. The function $f: 2^E \to \realset^+$ is a diversity function from subsets of the ground set $E$ to non-negative real numbers.

The diversity function $f$ can have many different forms. For instance, $f(X)$ can be the number of unique genres covered by movies $X$ recommended by a recommender system. Alternatively, $f(X)$ can be the average pairwise dissimilarity between a set of products $X$ recommended by a shopping website. In this work, we assume that the diversity function $f$ is \emph{monotonically increasing}:
\begin{align}
  \forall X \subseteq E, e \in E - X: f(X + e) - f(X) \geq 0,
\end{align}
the diversity of any set $X$ does not decrease when any item $e$ is added to this set. This assumption is quite natural. We also assume that $f(\emptyset) = 0$, the diversity of the empty set is zero. This assumption is without loss of generality. In particular, it can be always satisfied by subtracting $f(\emptyset)$ from $f$.

\subsection{Greedy Solution}
\label{sec:algorithm}

For a general monotonic function $f$, the optimization problem \eqref{eq:DUM} is NP-hard. However, when $f$ is submodular, the problem can be cast as finding a maximum-weight basis of a polymatroid \cite{edmonds70submodular} and can be solved greedily. We first present the greedy algorithm and then argue that it is optimal.

The pseudo-code of the greedy algorithm for \emph{diversity-weighted utility maximization ($\dum$)} is shown in Algorithm~\ref{alg:DUM}. The algorithm works as follows. First, the items $E$ are sorted in decreasing order according to their utility, $\bw(a^\ast_1) \geq \ldots \geq \bw(a^\ast_L)$, and placed into $A^\ast = (a^\ast_1, \dots, a^\ast_L)$. Then we examine the items in this order. When $g_{A^\ast}(a^\ast_k) > 0$, item $a^\ast_k$ is added to the list of recommended items $S$. When $g_{A^\ast}(a^\ast_k) = 0$, item $a^\ast_k$ is not added to $S$ because it does not contribute to the diversity of $S$. Finally, the algorithm returns the recommendation list $S$.

We illustrate $\dum$ on the second example in Section~\ref{sec:motivation}. In this example, $A^\ast = (1, 2, 5, 3, 4)$, and the diversity gains of movies $2$, $3$ and $4$ are zero due to the contribution of their preceding movies in the list. Therefore, these movies are not placed into the recommendation list, and $S = (1, 5)$.

$\dum$ has several notable properties. First, it is \emph{parameter-free}. That is, $\dum$ does not require any parameter tuning and therefore should be robust in practice. Second, $\dum$ is a greedy method and therefore is \emph{computationally efficient}. In particular, suppose that the diversity function $f$ is an oracle that can be queried in $O(1)$ time. Then the time complexity of $\dum$ is $O(L \log L)$, comparable to the complexity of sorting $L$ numbers. Finally, $\dum$ computes the \emph{optimal solution} to the optimization problem \eqref{eq:DUM}.

\begin{algorithm}[t]
  \caption{$\dum$: Diversity-Weighted Utility Maximization}
  \label{alg:DUM}
  \begin{algorithmic}
    \STATE {\bf Input:}
    \STATE \quad Ground set $E$
    \STATE \quad Weight vector $\bw$
    \STATE
    \STATE // Compute the maximum-weight basis of a polymatroid
    \STATE Let $a^\ast_1, \dots, a^\ast_L$ be an ordering of items $E$ such that:
    \STATE \quad $\bw(a^\ast_1) \geq \ldots \geq \bw(a^\ast_L)$
    \STATE $A^\ast \gets (a^\ast_1, \dots, a^\ast_L)$
    \STATE
    \STATE // Generate the list of recommended items $S$
    \STATE $S \gets ()$
    \FOR{$k = 1, \dots, L$}
      \STATE $g_{A^\ast}(a^\ast_k) \leftarrow f(A^\ast_{k - 1} + a^\ast_k) - f(A^\ast_{k - 1})$
      \IF{$(g_{A^\ast}(a^\ast_k) > 0)$}
        \STATE Append item $a^\ast_k$ to list $S$
      \ENDIF
    \ENDFOR
    \STATE
    \STATE {\bf Output:}
    \STATE \quad List of recommended items $S$
  \end{algorithmic}
\end{algorithm}

In the rest of Section~\ref{sec:model}, we analyze $\dum$ both in terms of $A^\ast$ and $S$. Note that the solutions $A^\ast$ and $S$ are equivalent in the sense that $S$ is a list obtained from $A^\ast$ by eliminating the items that have zero contribution in the objective function \eqref{eq:DUM}. Therefore, the values of the solutions are identical. So the difference in treatment is purely technical and allows us to reduce overhead in notation.

The optimality of $\dum$ can be argued based on the following observation. Our optimization problem \eqref{eq:DUM} is equivalent to maximizing a modular function on a polymatroid \cite{edmonds70submodular}, a well-known combinatorial optimization problem that can be solved greedily. In particular, let $M = (E, f)$ be a polymatroid, where $E$ is its ground set and $f$ is a submodular diversity function. Let:
\begin{align}
  & P_M = \label{eq:independence polyhedron} \\
  & \set{\bx: \bx \in \realset^L, \ \bx \geq 0,
  \ \forall X \subseteq E: \sum_{e \in X} \bx(e) \leq f(X)} \nonumber
\end{align}
be the independence polyhedron associated with function $f$. Then the maximum-weight basis of $M$ is defined as:
\begin{align}
  \bx^\ast = \argmax_{\bx \in P_M} \langle\bw, \bx\rangle,
  \label{eq:maximum-weight basis}
\end{align}
where $\bw \in (\realset^+)^L$ is a vector of non-negative weights. Because $P_M$ is a submodular polytope and the weights $\bw$ are non-negative, the optimization problem \eqref{eq:maximum-weight basis} is equivalent to finding the order of dimensions $A$ in which $\langle\bw, \bx\rangle$ is maximized \cite{edmonds70submodular}. This problem can be written formally as \eqref{eq:DUM} and has the same greedy solution as in $\dum$. In particular, the items $E$ are sorted in decreasing order according to their weights, $\bw(a^\ast_1) \geq \ldots \geq \bw(a^\ast_L)$, and placed into $A^\ast = (a^\ast_1, \dots, a^\ast_L)$. Finally, $\bx^\ast  = g_{A^\ast}$.

\subsection{Interpretation}
\label{sec:interpretation}

In this section, we discuss several interpretations of $\dum$. Without loss of generality, we assume that the different aspects of user's taste are represented by a finite set of topics $\cT = \set{1, \dots, M}$. For example, in a movie recommendation domain, these topics can be the genres of movies, such as $\cT = \set{\textit{Drama}, \textit{Comedy}, \textit{Action}}$.

Our first observation is that if the diversity of a set of items is measured by the number of unique topics covered by the items, then $\dum$ generates a list of items, where each topic is covered by the highest-utility item in that topic.

\begin{lemma}
\label{lem:topic diversity} Let the diversity function $f$ be defined as the number of topics covered by items $X$:
\begin{align*}
  f(X) = \sum_{t \in \cT} \I{\exists e \in X: \emph{item }e\emph{ covers topic }t}.
\end{align*}
Then $\dum$ returns a recommendation list $S$, where each topic $t$ is covered by the highest-utility item that belongs to $t$. Moreover, the length of $S$ is at most $\abs{\cT}$.
\end{lemma}
\begin{proof}
The first claim is proved by contradiction. Let $e^\ast_t$ be the item with the highest utility that belongs to topic $t$. Suppose that item $e^\ast_t$ is not chosen by $\dum$, $e^\ast_t$ is not in list $S$ generated by $\dum$. Then $g_{A^\ast}(e^\ast_t) = 0$, which implies that another item must have covered topic $t$ before item $e^\ast_t$. However, this is a contradiction, since $e^\ast_t$ is the item with the highest utility from $t$, and therefore $\dum$ must have tested it before any other item that covers $t$.

The second claim follows from the fact that $g_{A^\ast}(a^\ast_k) > 0$ implies that the value of the diversity function $f$ increases by at least one. By definition, $f(X) \leq \abs{\cT}$ for any $X$. Therefore, the maximum number of items added to $S$ is $\abs{\cT}$.
\end{proof}

Our second observation is that our objective \eqref{eq:DUM} can be viewed as maximizing the \emph{expected utility} of choosing an item when the diversity gains $g_A(e)$ are viewed as the probabilities of choosing items. This interpretation is motivated by the cascade model \cite{craswell2008experimental} of user behavior, which considers the relationship between successive items in a list. In this model, users scan the list from top to the bottom and eventually stop because either their information need is satisfied or their patience is exhausted.

Specifically, note that for any ordering of items $A$:
\begin{align}
  \sum_{k = 1}^L g_A(a_k)
  & = \sum_{k = 1}^L [f(A_{k - 1} + a_k) - f(A_{k - 1})] \nonumber \\
  & = f(E) - f(\emptyset) + \sum_{k = 1}^{L - 1} [f(A_k) - f(A_k)] \nonumber \\
  & = f(E).
\end{align}
The first equality is due to the definition of the diversity gains \eqref{eq:gain}. The second equality follows from the fact that $A_k = A_{k - 1} + a_k$. The last equality is due to the observation that $f(\emptyset) = 0$. It follows that:
\begin{align}
  \forall e \in E: \frac{g_A(e)}{f(E)} \in [0, 1], \quad
  \frac{1}{f(E)} \sum_{k = 1}^L g_A(a_k) = 1,
\end{align}
and therefore $g_A(e) / f(E)$ can be interpreted as the probability of choosing item $e$, given that none of the earlier recommended items $A_{k - 1}$ is chosen. Under this assumption, $\sum_{k = 1}^L g_A(a_k) \bw(a_k)$ is the expected utility of choosing an item, scaled up by $f(E)$.

\subsection{Diversity Function}
\label{sec:personalization}

The length of the recommendation list $S$ generated by $\dum$ depends on the diversity function $f$. In extreme cases, this list may include all items. For instance, consider a problem where:
\begin{align}
  \forall X \subseteq E, e \in E - X: f(X + e) - f(X) > 0,
\end{align}
the diversity increases when any item $e$ is added to any subset of items $X$. Then for any ordering $A$, $g_A(e) > 0$ for all items $e$. As a result, $g_{A^\ast}(e) > 0$ for all items $e$ and $\dum$ returns the list of all items sorted in the descending order of utility. This result is mathematically correct. But it is not a very useful diverse recommendation.

To get useful diverse recommendations, it is important to control the maximum number of items returned by $\dum$, e.g., by choosing appropriate diversity functions. For instance, for the diversity function in Lemma~\ref{lem:topic diversity}, the maximum length of the recommendation list $S$ is equal to the number of topics $\abs{\cT}$. In this section, we generalize the ideas from Section~\ref{sec:interpretation} and propose another class of diversity functions that are suitable for $\dum$.

Consider the case where different users may have different tolerance for redundancy in the recommendation list due to their interests and preferences~\cite{lad2010learning, vargas2012explicit}. These differences can be modeled by a diversity function that assigns different weights to each topic of interest. In particular, the function can be defined as:
\begin{align}
  f(X) = \sum_{t \in \cT} \min\set{\sum_{e \in X} \I{\text{item }e\text{ covers topic }t}, N_t},
  \label{eq:user diversity}
\end{align}
where $N_t$ is the number of items from topic $t$ that is required to be in the recommendation list of a given user. In the next lemma, we characterize the output of $\dum$ for the above function.

\begin{lemma}
\label{lem:user diversity} Let the diversity function $f$ be defined as in \eqref{eq:user diversity}. Then $\dum$ returns a recommendation list $S$ such that each topic $t$ is covered by at least $N_t$ items of the highest utility that cover topic $t$. Moreover, the length of $S$ is at most $\sum_{t \in \cT} N_t$.
\end{lemma}
\begin{proof}
The first claim is proved by contradiction. Let $e^\ast_{t, k}$ be the $k$-th item with the highest utility from topic $t$, where $k \leq N_t$. Suppose that item $e^\ast_{t, k}$ is not chosen by $\dum$, $e^\ast_{t, k}$ is not in list $S$ generated by $\dum$. Then $g_{A^\ast}(e^\ast_{t, k}) = 0$, which implies that topic $t$ must have been covered at least $N_t$ times before $\dum$ tests item $e^\ast_t$. However, note that this is a contradiction, since $e^\ast_{t, k}$ is among the first $N_t$ items that cover topic $t$, and therefore among the first $N_t$ items from that topic that are tested by $\dum$.

The second claim follows from the fact that $g_{A^\ast}(a^\ast_k) > 0$ implies that the value of the diversity function $f$ increases by at least one. By definition, $f(X) \leq \sum_{t \in \cT} N_t$ for any $X$. Therefore, the maximum number of items added to $S$ is $\sum_{t \in \cT} N_t$.
\end{proof}

The diversity function in \eqref{eq:user diversity} allows for controlling the length of the recommendation list $S$. In particular, if topic $t$ is irrelevant for the user, $N_t$ should be set to $0$. As a rule of thumb, more relevant topics $t$ should be assigned higher weights $N_t$.

%% file: Experiments.tex

\section{Experiments}
\label{sec:experiments}

The proposed method is evaluated in two online user studies and in an offline evaluation. In each case, we compare the performance of $\dum$ to variants of $\mmr$, since many existing diversification approaches are based on the objective function of $\mmr$ (Section~\ref{sec:related work}). In theory, the optimal solution to $\dum$ can be found greedily, while $\mmr$ finds only a $(1 - 1 / e)$-approximation to the optimal solution. Through the empirical evaluation, we show that $\dum$ satisfies the users' needs better than $\mmr$, and it is superior in recommending lists that satisfy utility and diversity at the same time.

We conduct two online studies using Amazon's Mechanical Turk\footnote{http://www.mturk.com} (MT). In the first study, we \emph{evaluate} separately the recommendation lists generated by $\dum$ and $\mmr$, by asking MT workers to identify in the lists a movie that matches their genre of interest and indicate the relevance of this movie. In the second study, we \emph{compare} the $\dum$ and $\mmr$ recommendation lists, by asking MT workers to judge  the coverage of two movie genres by the lists. We also report the findings of an offline study, where we perform a fine-grained assessment of the evaluated methods by creating user preference profiles and considering various combinations of genres.

%% file: UserStudy1.tex

\subsection{User Study 1}
\label{sec:study 1}

In the first study, we evaluate the diversity and utility of $\dum$ in a movie recommendation application. We compare $\dum$ to three variants of $\mmr$, which are parametrized by $\lambda \in \set{\frac{1}{3}, \frac{2}{3}, 0.99}$.

The ground set $E$ are $10$k most frequently rated IMDb\footnote{http://www.imdb.com} movies. The utility of movie $e$, $\bw(e)$, is the number of ratings assigned to this movie. The values of $\bw(e)$ are normalized such that the maximum utility is $1$, i.e., $\max_{e \in E} \bw(e) = 1$. The diversity function $f$ is defined as in Lemma~\ref{lem:topic diversity}. We also normalize $f$ such that the maximum diversity of is 1, i.e., $\max_{e \in E} f(e) = 1$. The set of topics $\cT$ includes 8 most popular movie genres in $E$:
\begin{align}
  \cT = \{&\textit{Drama}, \textit{Comedy}, \textit{Thriller}, \textit{Romance}, \\
  & \textit{Action}, \textit{Crime}, \textit{Adventure}, \textit{Horror}\}. \nonumber
\end{align}
We restrict our attention to $8$ genres only since $8$ genres can be always covered by $8$ movies, a reasonably short list of movies that can be evaluated by a MT worker.

\begin{figure}[t]
  \centering
  \includegraphics[width=3.45in]{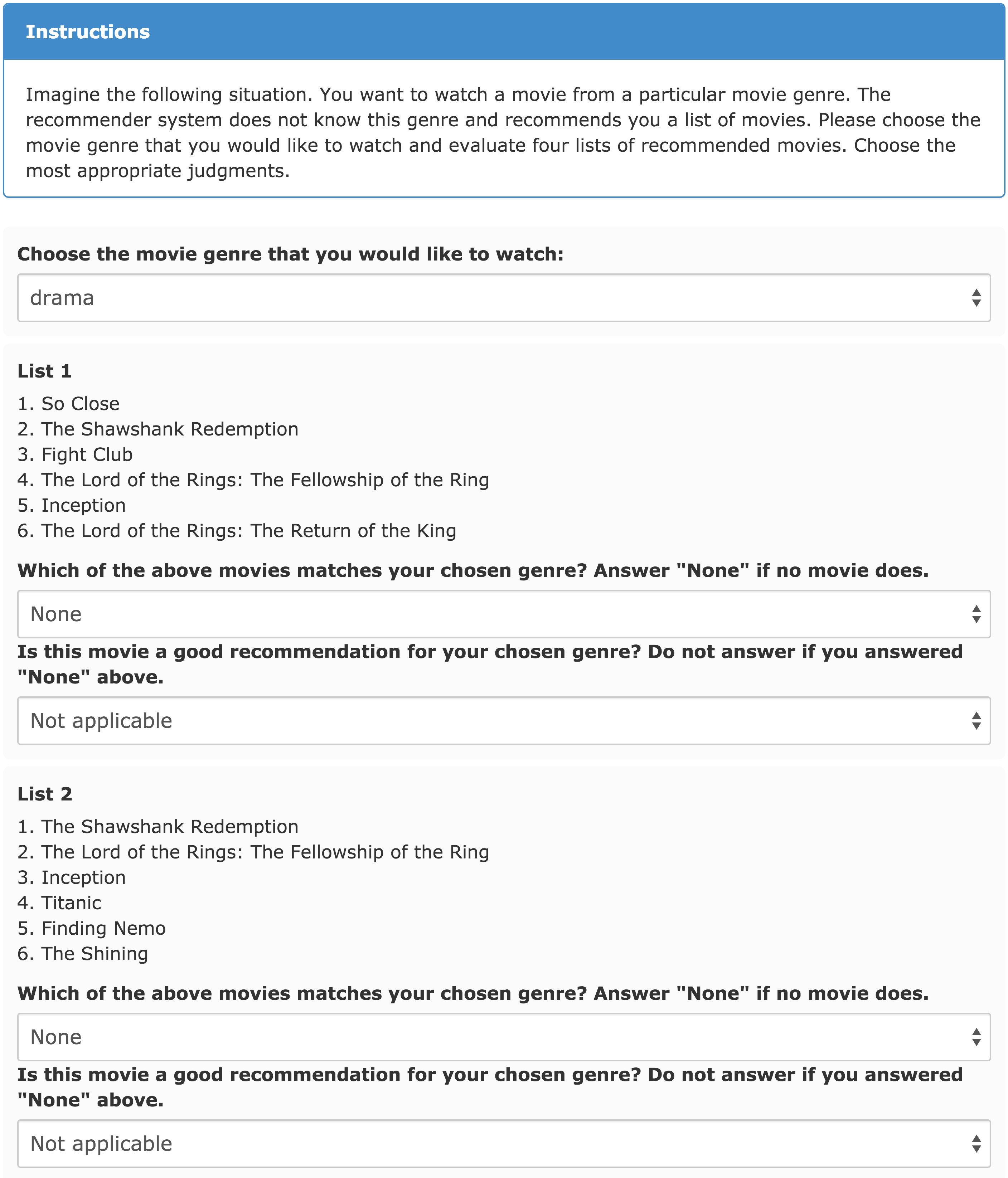}
  \caption{A portion of our Mechanical Turk questionnaire in user study $1$. We only show the first two lists of recommended movies.}
  \label{fig:questionnaire 1}
\end{figure}

All methods are evaluated in $200$ MT human intelligence tasks (HITs). In each HIT, we initially ask the worker to choose a genre of interest. Then, we generate four recommendation lists: one by $\dum$ and three by $\mmr$ for different values of $\lambda$. The generation of the lists is independent of the genre chosen by the worker. Finally, we ask the worker to evaluate the lists. For each list, we ask two questions. First, we ask the worker to identify a movie in the list that matches the chosen genre. This question addresses the diversity of the list, whether the chosen genre is covered by the list. The worker can also answer ``none'' if the list does not contain a movie from the chosen genre. If the worker identifies a matching movie, we ask the worker if the movie is a good recommendation for the chosen genre. This question addresses the utility of the list, whether the chosen genre is covered by a good movie in the list. A screenshot of our MT questionnaire is shown in Figure~\ref{fig:questionnaire 1}.

In each HIT, we present the four recommendation lists in a random order. This eliminates the \emph{position bias}. In addition, in each HIT the set of recommendable movies contains $3.3$k movies chosen at random from the $10$k movies in $E$. Hence, the recommendation lists differ across the HITs, which eliminates the \emph{item bias}, i.e., the workers cannot prefer one method over another only because the recommended movies are inherently more likable. All the recommendation lists are of the same length -- the length of the list produced by $\dum$. We adopt this methodology because we want to compare the lists for the same number of movies in the lists. Note that that we do not put $\mmr$ into any disadvantage. In particular, for any $\dum$ list, $\mmr$ can generate lists that are either of higher utility or more diverse than the $\dum$ list, when the value of $\lambda$ is large or small, respectively. This can be seen in Table~\ref{tab:example 1}, for instance.

Our HITs are completed by $34$ \emph{master} workers, who are MT's elite workers chosen based on the high quality of the work. Each worker is asked to complete at most $8$ HITs. This guarantees that our HITs are completed by more than just a handful of workers. On average, a worker spends $72$ seconds per HIT, i.e., $19$ seconds to evaluate a list of up to $8$ movies. Later in this section, we present two permutation tests that show that our results are highly unlikely under the hypothesis that the workers are of low quality, or that the questions are answered randomly. This implies that the workers have reasonable expertise in evaluating the HITs.

The results of our study are presented in Table~\ref{tab:study 1}. For each compared method, we report the percentage of times that the worker finds a movie in the list that matches the chosen genre and the percentage of times that the matching movie is a good recommendation. We observe two major trends.

\begin{table}[t]
  \centering
  {\small
  \begin{tabular}{lrrrr} \hline
    & $\dum$ & \multicolumn{3}{c}{$\mmr$} \\
    & & \multicolumn{1}{c}{$\lambda = \frac{1}{3}$} &
    \multicolumn{1}{c}{$\lambda = \frac{2}{3}$} &
    \multicolumn{1}{c}{$\lambda = 0.99$} \\ \hline
    List includes a movie that & 84.0\% & 70.5\% & 67.0\% & 66.5\% \\
    matches the chosen genre \\ 
    The chosen movie is & 77.0\% & 64.5\% & 62.5\% & 62.5\% \\
    a good recommendation \\ \hline
  \end{tabular}
  }
  \caption{Comparison of $\dum$ and $\mmr$ in user study $1$. For each method, we report the percentage of times that the worker finds a matching movie in the list and the percentage of times that the matching movie is a good recommendation.}
  \label{tab:study 1}
\end{table}

Firstly, the percentage of times that the worker finds a matching movie in the $\dum$ list is $13.5\%$ higher than in the list generated by the best performing baseline, $\mmr$ with $\lambda = \frac{1}{3}$. This result is statistically significant and we show it using a permutation test. The \emph{test statistic} is the difference in the percentage of times that the worker finds a matching movie in the lists generated by the best and second best performing methods. The \emph{null hypothesis} is that all compared methods are equally good. Under this hypothesis, we permute the answers of the workers $10^6$ times, generate an empirical distribution of the test statistic, and observe that the value of $13.5\%$ or higher is less likely than $0.0001$. So we reject the null hypothesis with $p < 0.0001$.

Secondly, the percentage of times that the worker considers the chosen movie to be a good genre-matching recommendation in the $\dum$ list is $12.5\%$ higher than in the list generated by the best performing baseline, $\mmr$ with $\lambda = \frac{1}{3}$. This result is statistically significant and we show it again using a permutation test. The \emph{test statistic} is the difference in the percentage of times that the worker finds a good recommendation in the lists generated by the best and second best performing methods. The \emph{null hypothesis} is that all compared methods are equally good. Under this hypothesis, we permute the answers of the workers $10^6$ times, generate an empirical distribution of the test statistic, and observe that the value of $12.5\%$ or higher is less likely than $0.001$. So we reject the null hypothesis with $p < 0.001$. Overall, this user study shows that the diversity and the utility of recommendation lists generated by $\dum$ are perceived superior to those of the lists generated by $\mmr$.

We note that for all the methods compared in Table~\ref{tab:study 1}, the ratio between the percentage of times that the genre-matching movie is a good recommendation and that the matching movie is found is always between $0.92$ and $0.94$. This implies that if a matching movie found, it is very likely to be considered a good recommendation. We conjecture that this is due to the high popularity of movies in the ground set $E$, which practically guarantees the utility of the recommended movies and minimizes the differences between the compared methods.

In Table~\ref{tab:example 1}, we show a real-life example illustrating how $\dum$ outperforms $\mmr$. Here, $\dum$ covers all the $8$ movie genres by popular movies. These movies are well known and can be easily matched to any chosen target genre. However, $\mmr$ with $\lambda = 0.99$ assigns insufficient weight to diversity and therefore covers only five movie genres. The result is that this $\mmr$ list is unsuitable for users who like \emph{Comedy}, \emph{Romance}, or \emph{Horror} movies. $\mmr$ with $\lambda = \frac{2}{3}$ has the same problem. On the other hand, $\mmr$ with $\lambda = \frac{1}{3}$ assigns too much weight to diversity and therefore covers four movie genres by a relatively unknown movie, ``Phibes Rises Again''. These genres are not covered by any other movie in the list. The result is that this $\mmr$ list is likely to be of a low utility for users who like \emph{Comedy}, \emph{Romance}, \emph{Adventure}, and \emph{Horror} movies.

\begin{table}[t]
  \centering
  {\small
  \begin{tabular}{ll} \hline \hline
    \multicolumn{2}{c}{$\dum$} \\ \hline
    The Shawshank Redemption & drama crime\\
    The Dark Knight & drama thriller action crime\\
    The Lord of the Rings 1 & action adventure\\
    Forrest Gump & drama romance\\
    Back to the Future & comedy adventure\\
    The Shining & drama horror \\ \hline \hline
    \multicolumn{2}{c}{$\mmr$ ($\lambda = 1 / 3$)} \\ \hline
    The Dark Knight & drama thriller action crime\\
    Dr. Phibes Rises Again & comedy romance adventure horror\\
    The Shawshank Redemption & drama crime\\
    Pulp Fiction & thriller crime\\
    Fight Club & drama\\
    The Godfather & drama crime \\ \hline \hline
    \multicolumn{2}{c}{$\mmr$ ($\lambda = 2 / 3$)} \\ \hline
    The Dark Knight & drama thriller action crime\\
    The Shawshank Redemption & drama crime\\
    The Lord of the Rings 1 & action adventure\\
    Pulp Fiction & thriller crime\\
    Fight Club & drama\\
    The Godfather & drama crime \\ \hline \hline
    \multicolumn{2}{c}{$\mmr$ ($\lambda = 0.99$)} \\ \hline
    The Shawshank Redemption & drama crime\\
    The Dark Knight & drama thriller action crime\\
    Pulp Fiction & thriller crime\\
    Fight Club & drama\\
    The Godfather & drama crime\\
    The Lord of the Rings 1 & action adventure \\ \hline \hline
  \end{tabular}
  }
  \caption{Four recommended lists in user study $1$ where $\dum$ outperforms $\mmr$.}
  \label{tab:example 1}
\end{table}

%% file: UserStudy2.tex

\subsection{User Study 2}
\label{sec:study 2}

In the second study, we evaluate $\dum$ on a specific problem of recommending a diverse set of movies that cover exactly two genres. We again compare $\dum$ to three variants of $\mmr$, which are parameterized by $\lambda \in \set{\frac{1}{3}, \frac{2}{3}, 0.99}$.

The compared methods are evaluated by MT workers. In each HIT, we ask the worker to consider a situation where Bob and Alice go for a vacation and can take several movies with them. Bob and Alice prefer two different movie genres. We generate four recommendation lists: one by $\dum$ and three by $\mmr$ for different values of $\lambda$. For each list, we ask the worker to indicate whether the list is appropriate for both Bob and Alice, only for one of them, or for none of them. A screenshot of our MT questionnaire is shown in Figure~\ref{fig:questionnaire 2}.

\begin{figure}[t]
  \centering
  \includegraphics[width=3.4in]{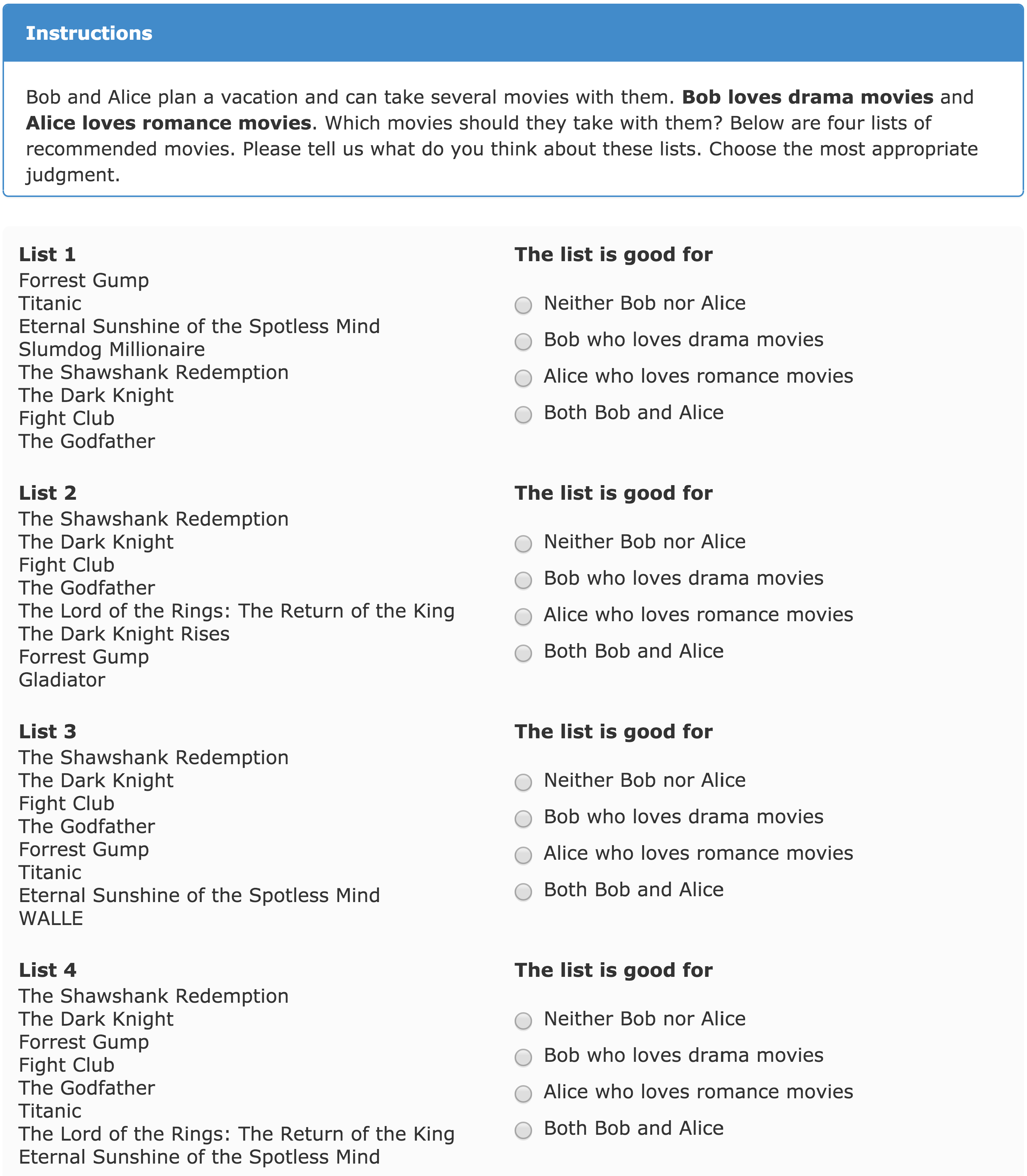}
  \caption{Our Mechanical Turk questionnaire in user study 2 for $t_1 = \textit{Drama}$ and $t_2 = \textit{Romance}$.}
  \label{fig:questionnaire 2}
\end{figure}

Each HIT is associated with two movie genres, $t_1$ and $t_2$, the preferences of Bob and Alice in the HIT. We generate three HITs for each pair of the 18 most frequent IMDb movie genres, so that the recommendation lists are evaluated $3 \frac{18 \times 17}{2} = 459$ times. Like in Section~\ref{sec:study 1}, the ground set $E$ are $10$k most frequently rated IMDb movies. The utility of movie $e$, $\bw(e)$, is the number of ratings assigned to this movie. The diversity function $f$ is defined as in Lemma~\ref{lem:user diversity}. The topics are $\cT = \set{t_1, t_2}$ and $N_{t_1} = N_{t_2} = 4$. For this setting, $\dum$ generates a list of at most 8 movies, at least 4 from each genre. The utility and diversity are normalized as in Section~\ref{sec:study 1}. In each HIT, the order of the recommendation lists is randomized and the length of the lists is determined as in Section~\ref{sec:study 1}.

\begin{table}[t]
  \centering
  \begin{tabular}{lrrrr} \hline
    Suitable & $\dum$ & \multicolumn{3}{c}{$\mmr$} \\
    for & & \multicolumn{1}{c}{$\lambda = \frac{1}{3}$} &
    \multicolumn{1}{c}{$\lambda = \frac{2}{3}$} &
    \multicolumn{1}{c}{$\lambda = 0.99$} \\ \hline
    Bob and Alice & 74.51\% & 64.92\% & 58.39\% & 28.98\% \\
    Bob or Alice & 23.53\% & 32.68\% & 39.43\% & 66.67\% \\
    Neither & 1.96\% & 2.40\% & 2.18\% & 4.36\% \\ \hline
  \end{tabular}
  \caption{Comparison of $\dum$ and $\mmr$ in user study 2. For each method, we report the percentage of times that the worker identifies the recommended list as suitable for both Bob and Alice; only for Bob or only for Alice; or for neither of them.}
  \label{tab:study 2}
\end{table}

Our HITs are completed by $57$ \emph{master} workers. Each worker is asked to complete at most $10$ HITs. This guarantees that our HITs are completed by more than just a handful of workers. On average, a worker spends $57$ seconds per HIT, i.e., $14$ seconds to evaluate a list of up to $8$ movies. In our analysis, we do not differentiate between suboptimal answers ``Suitable only for Alice'' and ``Suitable only for Bob'' and collapse the two into a single answer ``Suitable for Alice or Bob''. The results of the second user study are presented in Table~\ref{tab:study 2}.

We observe that the workers consider the $\dum$ list to be suitable for both Bob and Alice in $74.51\%$ of cases. This is $9.6\%$ higher than the best performing baseline, $\mmr$ with $\lambda = \frac{1}{3}$. This result is statistically significant and we show it using a permutation test. The \emph{test statistic} is the difference in the percentage of times that the recommended lists, generated by the best and second best performing methods, are suitable for both Bob and Alice. The \emph{null hypothesis} is that all compared methods are equally good. Under this hypothesis, we permute the answers of the workers $10^6$ times, generate an empirical distribution of the test statistic, and observe that the value of $9.6\%$ or higher is less likely than $0.0001$. So we reject the null hypothesis with $p < 0.0001$.

Similarly to Section~\ref{sec:study 1}, our permutation test can be also interpreted as showing that our results are highly unlikely under the hypothesis that the workers are of low quality, the lists are rated randomly. This implies that our workers have reasonable expertise in evaluating our HITs.

\begin{table}[t]
  \centering
  {\small
  \begin{tabular}{ll} \hline \hline
    \multicolumn{2}{c}{$\dum$} \\ \hline
    The Dark Knight & action \\
    The Lord of the Rings 1 & action \\
    The Matrix & action \\
    Inception & action \\
    The Shining & horror \\
    Alien & horror \\
    Psycho & horror \\
    Shaun of the Dead & horror \\ \hline \hline
    \multicolumn{2}{c}{$\mmr$ ($\lambda = 1 / 3$)} \\ \hline
    Zombieland & horror action \\
    From Dusk Till Dawn & horror action \\
    Dawn of the Dead & horror action \\
    Resident Evil & horror action \\
    The Dark Knight & action \\
    The Lord of the Rings 1 & action \\
    The Matrix & action \\
    Inception & action \\ \hline \hline
    \multicolumn{2}{c}{$\mmr$ ($\lambda = 2 / 3$)} \\ \hline
    The Dark Knight & action \\
    The Lord of the Rings 1 & action \\
    The Matrix & action \\
    Inception & action \\
    The Lord of the Rings 2 & action \\
    The Dark Knight Rises & action \\
    The Lord of the Rings 3 & action \\
    The Shining & horror \\ \hline \hline
    \multicolumn{2}{c}{$\mmr$ ($\lambda = 0.99$)} \\ \hline
    The Dark Knight & action \\
    The Lord of the Rings 1 & action \\
    The Matrix & action \\
    Inception & action \\
    The Lord of the Rings 2 & action \\
    The Dark Knight Rises & action \\
    The Lord of the Rings 3 & action \\
    Avatar & action \\ \hline \hline
  \end{tabular}
  }
  \caption{Four recommended lists in user study $2$ where $\dum$ outperforms $\mmr$. The topics are $t_1 = \textit{Horror}$ and $t_2 = \textit{Action}$.}
  \label{tab:example 2}
\end{table}

In Table~\ref{tab:example 2}, we show another real-life example illustrating how $\dum$ outperforms $\mmr$ for $t_1 = \textit{Horror}$ and $t_2 = \textit{Action}$. Here, $\dum$ covers each movie genre by four most popular movies from that genre. However, $\mmr$ with $\lambda = 0.99$ assigns insufficient weight to diversity and therefore recommends only most popular items that happen to be \emph{Action} movies. So the recommendation list is unsuitable for users who like \emph{Horror} movies. $\mmr$ with $\lambda = \frac{2}{3}$ has a similar behavior and is strongly dominated by \emph{Horror} movies. On the other hand, $\mmr$ with $\lambda = \frac{1}{3}$ assigns too much weight to diversity and therefore recommends many $\emph{Horror}$ movies that are also $\emph{Action}$ movies. These are less popular than the most popular $\emph{Horror}$ movies that are not $\emph{Action}$. So the list is of a low utility for users who like \emph{Horror} movies.

To sum up, $\dum$ outperforms $\mmr$ in cases, where items from one topic have a higher utility than items from the other topic, and the items at the intersection of the two topics also have a low utility. While $\dum$ recommends a mixture of high utility items from each topic, $\mmr$ either prefers items at the intersection of the topics, when the value of $\lambda$ is low; or recommends high-utility items from the dominant topic only, when the value of $\lambda$ is high.

%% file: Offline.tex

\subsection{Offline Evaluation}
\label{sec:offline}

The main goal of the offline evaluation is to assess the performance of $\dum$ under various conditions, such as recommendations across multiple users with their interest profiles defined based on different combinations of topics.

We use the \emph{1M MovieLens} dataset \cite{movielens} in the offline evaluation. The dataset consists of 1 million ratings on a 1-to-5 stars scale, assigned by about 6000 users to about 4000 movies. We remove users having less than 300 ratings, so that for each user we have enough data to create a user profile and recommend movies. We end up with 1000 users and a total of 515k ratings.

Movies rated by each user are split randomly into the training and test set with the $2\colon1$ ratio.
We use matrix factorization~\cite{thurau2011convex} to predict the rating of movies in the test set and feed the predicted ratings as the utility scores into the $\dum$ and $\mmr$ methods. The split is performed three times for each user, and the reported results are based on the average of experiments conducted on these splits.

For each user, the training set is used for creating their interest profile, whereas the test set contains the recommendable movies (along with their actual and predicted utility). On average, a user profile is created based on 343 movies and the recommendation list is selected from a set of 171 candidates. These steps are carried out as follows:

\textbf{User profile creation:} There are 18 genres of movies in the dataset, and each movie belongs to one or more of these genres. For each user, we create a multinomial distribution over the popularity of genres of the movies rated by this user in the training data, assuming that users rated movies that they had watched. We sample 10 times from this distribution to create the user's preference profile over genres, and normalize it so that the sum of the scores equals 1. For each user with the preference score $r_t$ for genre $t$, we set $N_t = \lfloor r_t \times K\rfloor$ in~\eqref{eq:user diversity}, where $K$ is the length of the recommendation list. That is, the coverage of each genre in the result list is proportional to the preference of the user for that genre.

\begin{figure*}
  \centering
  \includegraphics[width=2.8in]{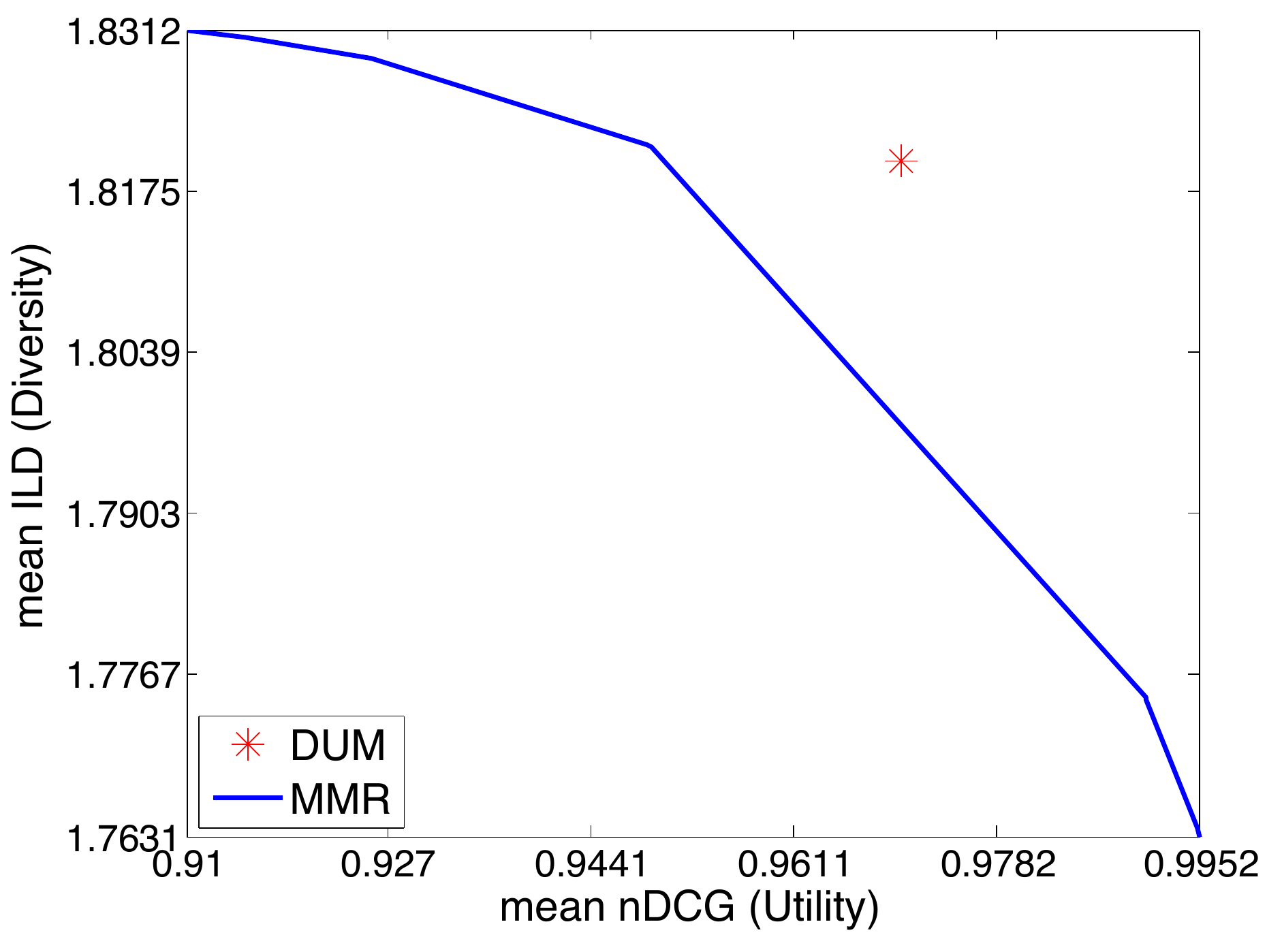} \hspace{+0.9in}
   \includegraphics[width=2.8in]{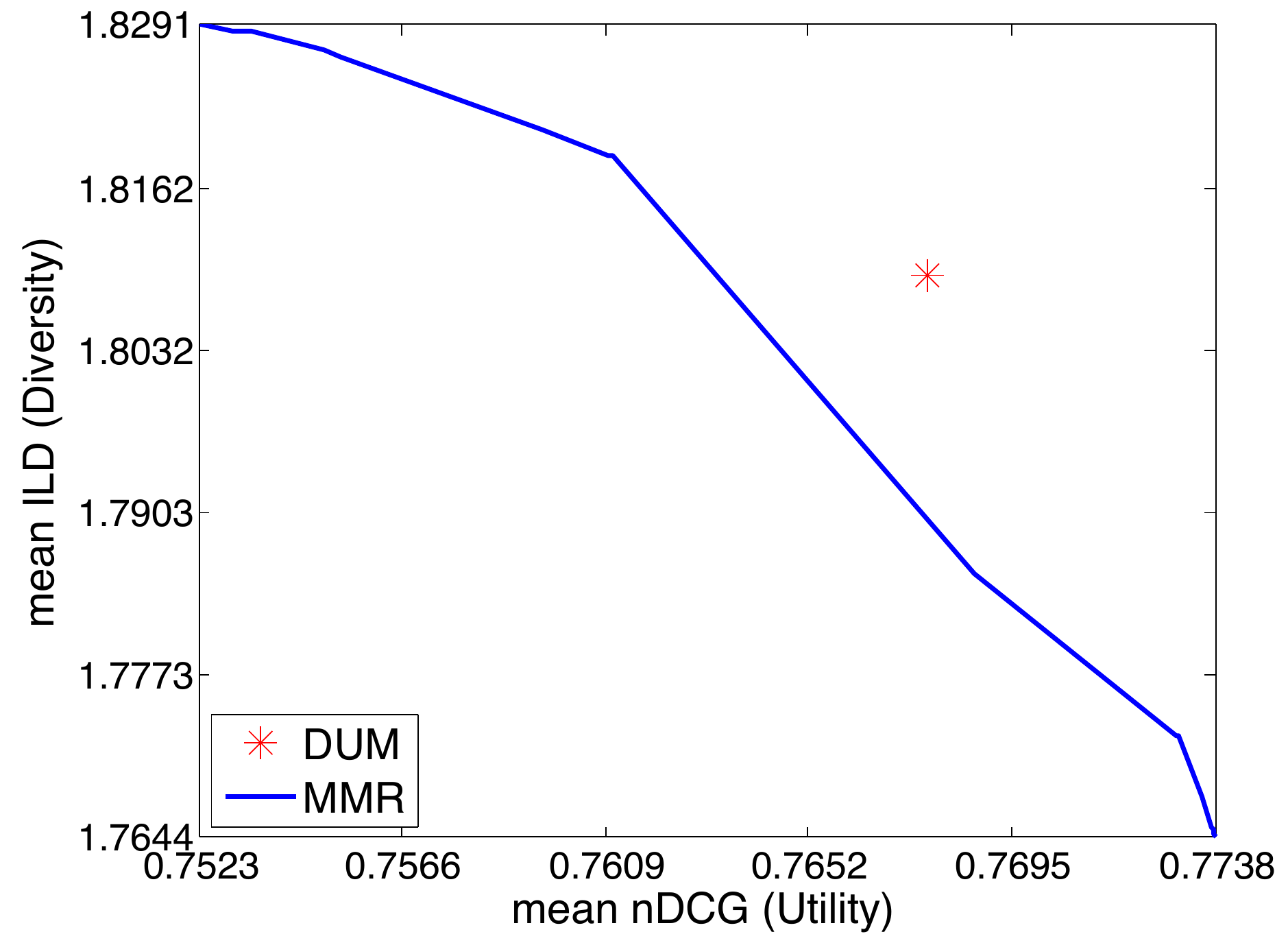}
  \text{\hspace{+0.1in}	\textbf{(a)}		\hspace{+3.6in}		\textbf{(b)}   }
  \caption{Performance of $\dum$ in terms of diversity and utility is compared to the performance of $\mmr$ for all the settings of the parameter $\lambda$. (a) The actual rating of movies is the utility score, (b) The predicted rating of movies is the utility score.}
  \label{fig:diversity-utility}
\end{figure*}

\textbf{Recommendation:} Movies in the test data are used as the ground set $E$ of recommendable movies, from which each diversification method finds the list of $K=10$ movies to recommend to each user. The predicted utility of movies is used in the recommendation. The reason for using the predicted utility instead of the readily available movie ratings is to keep the evaluation as close as possible to real-world recommendation scenarios, where the utility of items is not known. When we evaluate the performance of the studied methods, we use the actual utility score, i.e., the rating assigned by a user to a test set movie.

\begin{figure}
  \centering
  \includegraphics[width=3.3in]{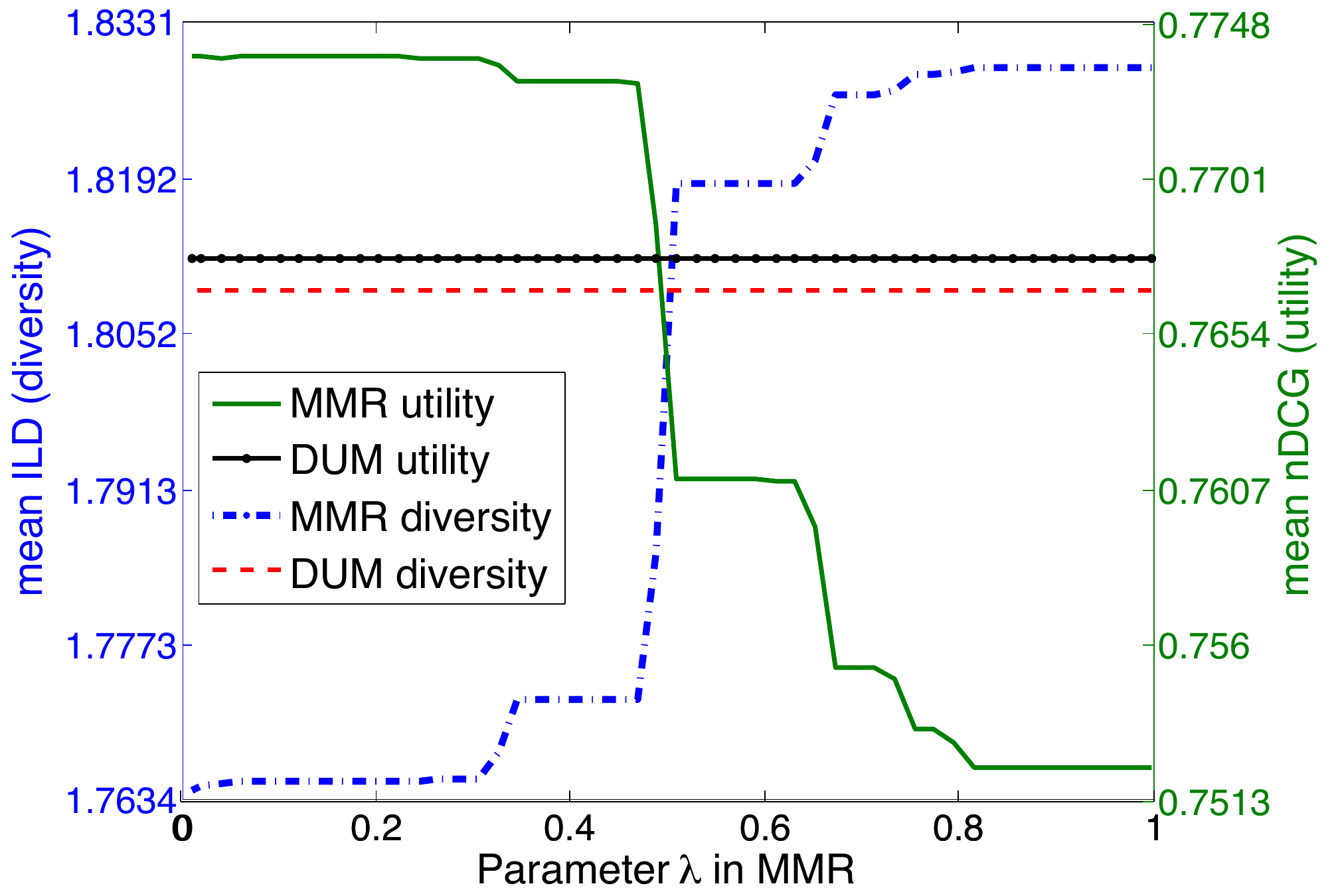}
  \caption{Tradeoff between the diversity and utility of the recommendation lists across all users for all the settings of the parameter $\lambda$ in $\mmr$. Diversity and utility scores achieved by $\dum$ (independent of $\lambda$) are shown for comparison purposes.}
\label{fig:diversity-utility-lambda}
\end{figure}

\subsubsection{Evaluation Metrics}

We use three evaluation metrics to compare the performance of $\dum$ to various settings of $\mmr$: a diversity metric, a utility metric, and a compound metric that considers both diversity and utility. We chose these particular metrics due to two reasons. First,
we wanted to evaluate the performance of our method with respect to diversity and utility individually (first two metrics), as well as in combination (third metric). Second, we wanted them to be different from the objective function of $\dum$ in order to avoid any potential bias. Thus, the compound metric combines diversity and utility in a different manner from what $\dum$ does.

\textbf{Intra-list distance (ILD)}~\cite{vargas2011rank, zhang2008avoiding} is a common metric that measures the diversity of a recommendation list as the average distance between pairs of recommended items. The dual of this measure is the intra-list similarity~\cite{ziegler2005improving}. We use ILD to measure distance based diversity of a recommendation list in our experiment:
\begin{align}
	ILD = \frac{2}{|S|(|S|-1)}  \sum_{e \in S} \sum_{e' \in S} \mathrm{d}(e, e')
	\label{eq:ILD}
\end{align}
where $\mathrm{d}(e, e')$ measures the distance between two items $e$ and $e'$ in a list $S$. We choose the Euclidean distance between the genre vectors of two movies as the distance function $d$. Note that this metric is cardinally different from the diversity function exploited by $\dum$, which is shown in~\eqref{eq:user diversity}.

\textbf{Discounted cumulative gain (DCG)}~\cite{jarvelin2002cumulated} measures the accumulated utility gain of items in the recommendations list from the top to the bottom, with the gain of each item $s_k$ being discounted by its position $k$ in the list:
\begin{align}
	DCG = \sum_{k=1}^{|S|} \frac{\bw(s_k)}{\log(k+1)}
	\label{eq:DCG}
\end{align}
Here, $\bw(s_k)$ is the utility of item $s_k$ at rank $k$ in the list. We estimate the utility of a movie for a user by the rating that the user assigned to the movie. We also use the normalized DCG (nDCG), which is in the range [0, 1]. nDCG is computed as nDCG = DCG/IDCG, where IDCG is the ideal gain achievable when all the listed items have the highest utility score.

\textbf{Expected intra-list distance (EILD)}~\cite{vargas2011rank} is a compound metric that combines utility and diversity. EILD measures the average intra-list distance (ILD) with respect to rank-sensitivity and utility:
\begin{align}
	EILD = \sum_{k = 1}^{|S|} \sum_{k' = 1}^{|S|} C_k \mathrm{disc}(k) \mathrm{rdisc}(k'|k) \bw(s_k) \bw(s_{k'}) \mathrm{d}(s_k, s_{k'})
	\label{eq:EILD}
\end{align}
where $\mathrm{disc}(k)=1/\log(k+1)$ is the discount function at rank $k$ in the list and $\mathrm{rdisc}(k'|k)=\mathrm{disc}(\max(1, k'-k))$ is the relative rank discount. In order to avoid bias, we use the normalization constant proposed in~\cite{vargas2011rank} and set $C_k = \frac{1}{C} / \sum_{k'=1}^{|S|} \mathrm{disc}(k'|k)\bw(s_{k'})$ where $C=\sum_{k=1}^{|S|}\mathrm{disc}(k)$.

We compute each of these metrics for every recommendation list provided to a user. Then, we average them across the three runs for every user to compute user-based mean of the metric. This is performed for $\dum$ and all settings of $\mmr$, and the mean of each metric for each method is computed across all the users and reported.

\subsubsection{Evaluation Results}

Figure~\ref{fig:diversity-utility} shows the performance of $\dum$ against $\mmr$ in terms of diversity and utility metrics. In Figure~\ref{fig:diversity-utility}-a, the actual rating of movies in the test set is used as the utility of movies in the recommendation step, whereas in Figure~\ref{fig:diversity-utility}-b the prediction produced by matrix factorization is used as the utility score. In both figures, $\mmr$ exhibits a trade-off between the values of mean ILD (as a measure of diversity) and mean nDCG (as a measure of utility). This trade-off is due to different values of the tuning parameter $\lambda$ in different settings of $\mmr$. For low values of $\lambda$ the utility is prioritized, such that the diversity of lists generated by $\mmr$ is low, but the utility is high. An opposite situation is observed for high $\lambda$, when the diversity gets prioritized.

It can be seen that the performance of $\dum$ with respect to both metrics is superior to any settings of $\mmr$, regardless of the way the utility score is obtained. For instance, in Figures~\ref{fig:diversity-utility}-b, $\dum$ achieves nDCG of 0.767 (compared to the highest nDCG of 0.774 achieved by $\mmr$ for $\lambda=0$) and ILD of 1.811 (compared to the highest ILD of 1.829 achieved by $\mmr$ for $\lambda=1$). It should be highlighted that the utility and diversity cannot be optimized by $\mmr$ simultaneously, since they are achieved for different values of $\lambda$, while $\dum$ competes with both of them at the same time. Also recall that $\dum$ is parameter-free, and its superiority over $\mmr$ becomes clear.

Comparing Figures~\ref{fig:diversity-utility}-a and \ref{fig:diversity-utility}-b, we observe that, as expected, the exact knowledge of the utility improves the performance of both $\dum$ and $\mmr$. However, this knowledge is unavailable to practical recommenders. Hence, we use the predicted utility values in the recommendation step in the following experiments, in order to mimic the conditions of a real-world recommendation scenario.

The trade-off between the diversity and utility objectives in $\mmr$ for various values of $\lambda$ is visualized in Figure~\ref{fig:diversity-utility-lambda}. As $\lambda$ increases, the diversity of the list recommended by $\mmr$ increases whereas its utility decreases. It can be seen that $\lambda=0.49$ is the operating point for $\mmr$, where the utility and the diversity curves intersect. On the contrary, the utility and diversity of $\dum$ are stable and both are above the operating point of $\mmr$.

Another argument in favor of $\dum$ is obtained through the EILD metric that combines diversity and utility. A comparison between $\dum$ and all the possible settings of $\mmr$ with respect to EILD are plotted in Figure~\ref{fig:ediversity-lambda}. It can be seen that $\dum$ significantly outperforms $\mmr$ for low/moderate values of $\lambda$, which correspond to cases where utility is prioritized, or both utility and diversity are similarly important. $\mmr$ starts outperforming $\dum$ for $\lambda \textgreater 0.65$, when the importance of diversity takes over, which may not be a favorite objective in real-world recommnedations. This confirms the superiority of $\dum$ in balancing the utility and diversity goals with no prior parameterization, compared to a method that explicitly targets the maximization of their weighted combination.

\begin{figure}
  \centering
  \includegraphics[width=3in]{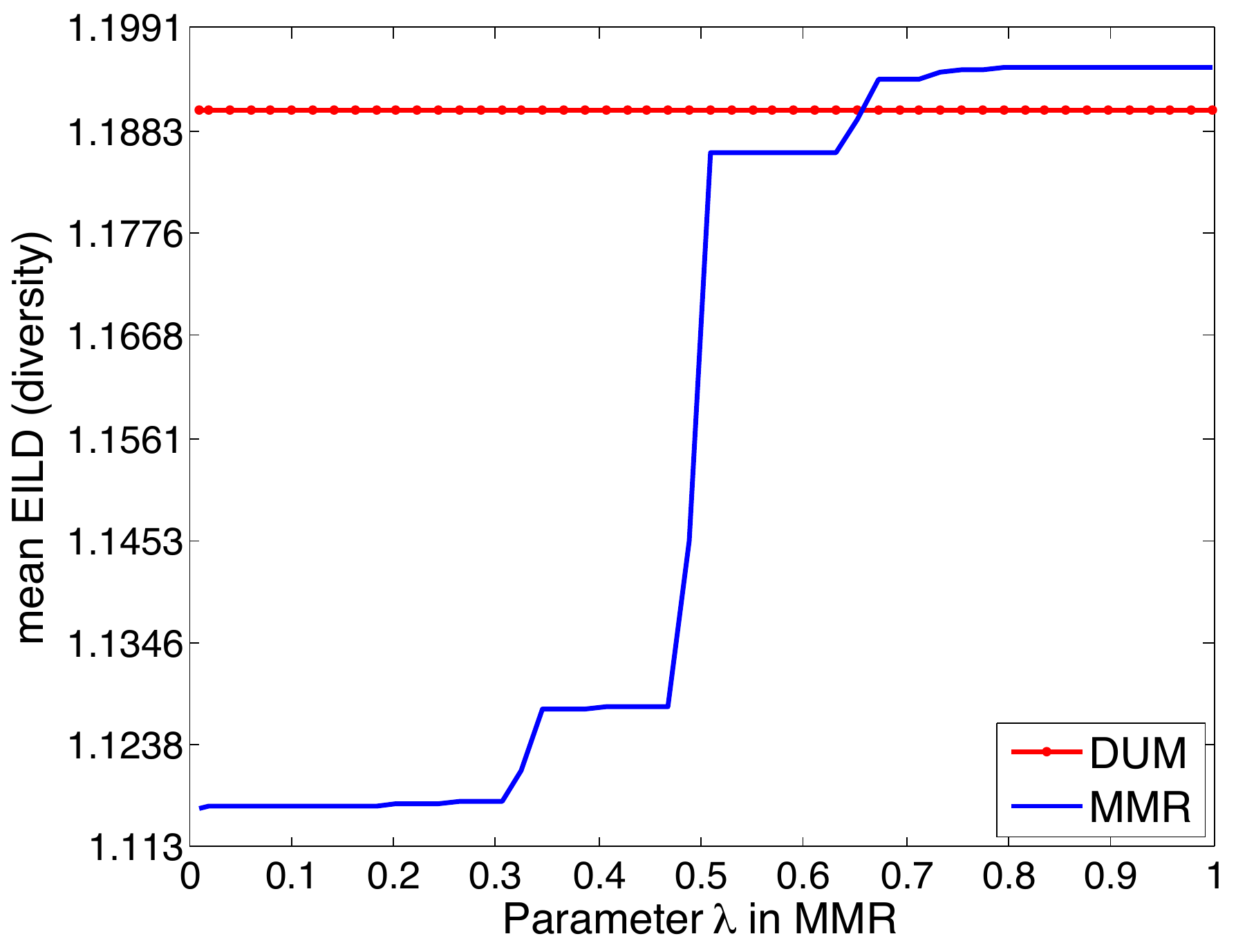}
  \caption{Performance of $\dum$ in terms of expected diversity (w.r.t. utility and rank) is compared to the performance of $\mmr$ for all the settings of the parameter $\lambda$.}
  \label{fig:ediversity-lambda}
\end{figure}

%% file: Conclusion.tex

\section{Conclusion}
\label{sec:conclusion}

Much research in recommender systems has focused on the accuracy, but overlooked issues related to the composition of the recommendation lists. Increasing the diversity of the lists poses a trade-off to the utility, such that the problem of maximizing utility subject to diversity is an important challenge. In this work, we propose the diversity-weighted utility maximization ($\dum$) method and show that the problem can be cast as finding the maximum of a modular function on a polymatroid, which is known to have an optimal greedy solution. This parameter-free method guarantees that items in the recommendation list cover different aspects of user's taste, such that each aspect is covered by items with high utility.


We conduct two online user studies. The diversity and utility of $\dum$ are evaluated in a movie recommendation scenario, and the perceived diversity of $\dum$ is evaluated in a specific problem of recommending a diverse set of movies that cover exactly two genres. In both studies, we found that $\dum$ outperforms baseline models that maximize a linear combination of utility and diversity. We also report an offline evaluation of $\dum$ using a suite of diversity and utility metrics. The results show that $\dum$ effectively balances the trade-off between diversity and utility: our method achieves performance comparable to the best performing baselines of diversity and utility, if executed individually. Moreover, a combined metric of diversity and utility shows the superiority of parameter-free $\dum$ over the baseline methods that need to be parameterized.

Most diversification methods use $\mmr$ objective function, to linearly combine modular and submodular functions of utility and diversity, respectively. Our work is orthogonal to these methods in the sense that the $\dum$ objective function maximizes a modular function subject to a submodular constraint. We demonstrate significant improvements over various settings of $\mmr$, while we intend to conduct a more encompassing comparison with other variants $\mmr$ in the future. Another future direction is to account for the novelty of the recommended items ~\cite{castells2011novelty,clarke2008novelty} with respect to prior consumption history of the user. This may be incorporated into the diversity function by considering, apart from the diversity contribution, also the novelty contribution of items in the list.

Another issue that deserves further investigation is the changes that need to be introduced in the diversity metric and in the tolerance for redundancy across different domains and applications. For instance, a metric of diversity applicable for news filtering may differ substantially from the metric we derived for the movie recommendation task in this work. Furthermore, user's tolerance for redundancy of news items that are in agreement with their own opinion may differ from their tolerance for redundancy of items having an opposite opinion. We intend to address these questions in our future works.

\vspace{+8mm}